\documentclass[11pt,onecolumn]{article}
\usepackage[margin=1in,nohead]{geometry}
\usepackage{microtype}
\usepackage{geometry}
\usepackage{balance}
\usepackage{amsmath}
\usepackage{amssymb}
\usepackage{amsthm}
\usepackage{multicol}
\usepackage{caption}
\usepackage{subcaption}
\usepackage{booktabs}
\usepackage{url}
\usepackage{algorithm}
\usepackage[noend]{algpseudocode}
\newtheorem{theorem}{Theorem}
\newtheorem{lemma}{Lemma}

\newtheorem{proposition}{Proposition}
\newtheorem{observation}{Observation}
\newtheorem{definition}{Definition}
\newtheorem{claim}{Claim}
\newtheorem{conjecture}{Conjecture}
\usepackage{tikz}
\usetikzlibrary{arrows,positioning,trees,matrix}
\pdfpageattr {/Group << /S /Transparency /I true /CS /DeviceRGB>>}
\tikzset{vertex/.style={circle,draw=black,fill=none,thin,node distance=0.5}}
\tikzset{edge/.style={thick}}
\algnewcommand{\LComment}[1]{\Statex \(\triangleright\) #1}

\title{\Large Graph connectivity in log steps using label propagation}
\author{Paul Burkhardt\footnote{Research Directorate, National Security Agency,
   Fort~Meade, MD 20755. Email: pburkha@nsa.gov}}
\date{December 07, 2021}

\begin{document}

\maketitle
\thispagestyle{empty}
\begin{abstract}
  The fastest deterministic algorithms for connected components take logarithmic
  time and perform superlinear work on a Parallel Random Access Machine
  (PRAM). These algorithms maintain a spanning forest by merging and compressing
  trees, which requires pointer-chasing operations that increase memory access
  latency and are limited to shared-memory systems. Many of these PRAM
  algorithms are also very complicated to implement. Another popular method is
  ``leader-contraction'' where the challenge is to select a constant fraction of
  leaders that are adjacent to a constant fraction of non-leaders with high
  probability, but this can require adding more edges than were in the original
  graph. Instead we investigate label propagation because it is deterministic,
  easy to implement, and does not rely on pointer-chasing. Label propagation
  exchanges representative labels within a component using simple graph
  traversal, but it is inherently difficult to complete in a sublinear number of
  steps. We are able to overcome the problems with label propagation for graph
  connectivity.

  We introduce a surprisingly simple framework for deterministic, undirected
  graph connectivity using label propagation that is easily adaptable to many
  computational models. It achieves logarithmic convergence independently of the
  number of processors and without increasing the edge count. We employ a novel
  method of propagating directed edges in alternating direction while performing
  minimum reduction on vertex labels. We present new algorithms in PRAM, Stream,
  and MapReduce. Given a simple, undirected graph $G=(V,E)$ with $n=|V|$
  vertices, $m=|E|$ edges, our approach takes O(m) work each step, but we can
  only prove logarithmic convergence on a path graph. It was conjectured by Liu
  and Tarjan (2019) to take $O(\log n)$ steps or possibly $O(\log^2 n)$
  steps. Our experiments on a range of difficult graphs also suggest logarithmic
  convergence. We leave the proof of convergence as an open problem.

\end{abstract}

\section{Introduction}
Given a simple, undirected graph $G=(V,E)$ with $n=|V|$ vertices and $m=|E|$
edges, the connected components of $G$ are partitions of $V$ such that every
pair of vertices are connected by a path, which is a sequence of adjacent edges
in $E$. If two vertices are not connected then they are in different
components. The distance $d(v,u)$ is the shortest-path length between vertices
$v$ and $u$. The diameter $D$ is the maximum distance in $G$. We wish to find
the connected components of $G$ in $O(\log n)$ steps using simple, deterministic
methods that are adaptable to many computational models.

The fastest, deterministic parallel ($\mathcal{NC}$) algorithms for connected
components take logarithmic time and perform superlinear work on a Parallel
Random Access Machine (PRAM). These algorithms maintain a spanning forest by
merging and compressing trees~\cite{bib:iwama1994,
  bib:colevishkin1991,bib:awerbuchshiloach1987, bib:shiloachvishkin1982}, which
requires pointer-chasing operations that increase memory access latency. Pointer
jumping was a primary source of slowdown in a parallel minimum spanning tree
algorithm~\cite{bib:chung_condon1996}. The PRAM implementations are also limited
to shared-memory
systems~\cite{bib:goddard1995,bib:patwary2012,bib:badercong2004}. Another
popular method is ``leader-contraction'' where the challenge is to select a
constant fraction of leaders that are adjacent to a constant fraction of
non-leaders with high
probability~\cite{bib:andoni2018,bib:kiveris2014,bib:rastogi2013,bib:karger1999}.
Not only is this method randomized but it can require adding many more edges to
the graph. Instead we investigate label propagation because it is deterministic,
easy to implement, and does not rely on pointer-chasing. Label propagation
exchanges representative labels within a component using simple graph traversal,
but it is inherently difficult to complete in a sublinear number of
steps~\cite{bib:rosenfeldpfaltz1966,bib:samet1988,bib:shapiro1996}. Adding and
removing edges must be carefully managed to keep the edge count, and therefore
the work, linear in each step. We are able to overcome the problems with label
propagation for graph connectivity.

We introduce a surprisingly simple framework for deterministic, undirected graph
connectivity using label propagation that is easily adaptable to many
computational models. It achieves logarithmic convergence independently of the
number of processors and without increasing the edge count. We employ a novel
method of propagating directed edges in alternating direction while performing
minimum reduction on vertex labels. We believe our solution to obtaining
sublinear convergence and near optimal work for connected components is one of
the simplest to date. Moreover, our experiments demonstrate fast convergence.

In this paper we say that a $(v,u)$ edge is directed from $v$ to $u$ and
$(\cdot, \cdot)$ denotes an ordered pair that distinguishes $(v,u)$ from
$(u,v)$. We call the counter-oriented $(v,u),(u,v)$ edges the \emph{twins} of a
conjugate pair. An undirected $\{u,v\}$ edge in $G$ is then comprised of these
twins. To propagate a label $w$ from edge $(v,u)$ we create just the $(u,w)$
twin. In the next step we reverse the direction to return the opposite twin,
$(w,u)$, if $w$ is the minimum label for $u$. We'll call these two edge
operations \emph{label propagation} and \emph{symmetrization},
respectively. Concomitant with label propagation is a \emph{min update} on $u$'s
minimum label, which may or may not change. Contraction of the graph is due to
the \emph{label propagation} operation, and \emph{symmetrization} ensures that a
vertex and its minimum label are able to exchange new minimum labels. Thus every
vertex in each step propagates and retains its current minimum label. Let $l(v)$
be the current minimum label for $v$. Then for an edge $(v,u)$ we get $(u,l(v))$
or $(u,v)$ in a single step, due to either \emph{label propagation} or
\emph{symmetrization}, respectively. Each edge is replaced by a new edge and
hence the method maintains a stable edge count. The essential operations are
summarized as follows.

\begin{itemize}
\item For every $(v,u)$ edge if $u$ is not $l(v)$ then \emph{min update} and
  \emph{label propagation}, else \emph{symmetrization}, repeating until no label
  changes.
\end{itemize}

\noindent
Figure~\ref{fig:method} illustrates our method where the starting minimum label
for each vertex is the lowest vertex ID among its neighbors and itself, and
eventually the graph is transformed into a star whose root is the component
label.

\renewcommand{\thesubfigure}{\arabic{subfigure}}
\begin{figure*}[t]
\captionsetup[subfigure]{position=b}
\small
\centering
\begin{subfigure}[t]{.25\textwidth}
  \centering
  \raisebox{8mm}{
  \begin{tikzpicture}
    \node [vertex,node distance=0.25] (1) {1};
    \node [vertex,node distance=0.25] (3) [right=of 1] {3};
    \node [vertex,node distance=0.25] (4) [right=of 3] {4};
    \node [vertex,node distance=0.25] (2) [right=of 4] {2};
    \draw [edge] (1) to (3);
    \draw [edge] (3) to (4);
    \draw [edge] (4) to (2);
  \end{tikzpicture}
  }
  \caption{}
\end{subfigure}
\begin{subfigure}[t]{.2\textwidth}
  \centering
  \begin{tikzpicture}
    \node [vertex] (3) {3};
    \node [vertex] (1) [above right=of 3] {1};
    \node [vertex] (4) [below right=of 1] {4};
    \node [vertex] (2) [below right=of 3] {2};
    \draw [edge] (3) to (1);
    \draw [edge] (4) to (2);
    \draw [edge,>=latex,->] (3) to (2);
    \draw [edge,>=latex,->] (4) to (1);
  \end{tikzpicture}
  \caption{}
\end{subfigure}
\begin{subfigure}[t]{.2\textwidth}
  \centering
  \begin{tikzpicture}
    \node [vertex] (3) {3};
    \node [vertex] (1) [above right=of 3] {1};
    \node [vertex] (4) [below right=of 1] {4};
    \node [vertex] (2) [below right=of 3] {2};
    \draw [edge] (1) to (3);
    \draw [edge,>=latex,->] (1) to (4);
    \draw [edge,>=latex,->] (4) to (2);
    \draw [edge,>=latex,->] (2.75) to (1.285);
    \draw [edge,>=latex,->] (2.105) to (1.255);
  \end{tikzpicture}
  \caption{}
\end{subfigure}
\begin{subfigure}[t]{.2\textwidth}
  \centering
  \raisebox{4mm}{
  \begin{tikzpicture}
    \node [vertex] (3) {3};
    \node [vertex] (1) [above right=of 3] {1};
    \node [vertex] (4) [below right=of 1] {4};
    \node [vertex] (2) [below=of 1] {2};
    \draw [edge] (3) to (1);
    \draw [edge] (2.75) to (1.285);
    \draw [edge,>=latex,->] (1.255) to (2.105);
    \draw [edge,>=latex,->] (4) to (1);
  \end{tikzpicture}
  }
  \caption{}
\end{subfigure}
\caption{\label{fig:method}A path graph converges in three steps. After each
  step the output graph becomes the input for the next step. Undirected edges
  denote a pair of counter-oriented edges.}
\end{figure*}
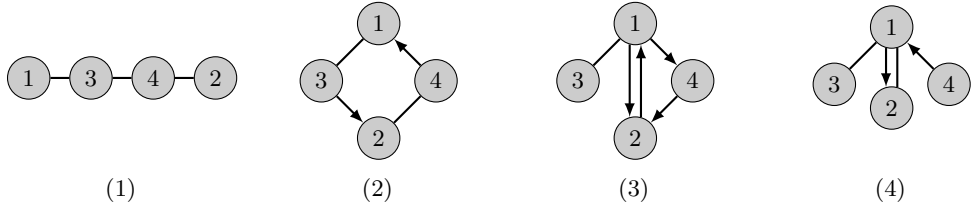
\renewcommand{\thesubfigure}{\alph{subfigure}}

This is a practical and very simple technique for large-scale streaming and
parallel graph connectivity. We present new algorithms in PRAM, Stream, and
MapReduce. Our approach takes $O(m)$ work each step, but we can only prove
logarithmic convergence on a path graph. Despite the simplicity of our
algorithm, the proof of logarithmic convergence is elusive and poses a rather
interesting challenge. We conjecture that our algorithm takes $O(\log n)$ steps
to converge. Our algorithm behaves well and empirically takes $O(\log n)$ steps
on a range of difficult graphs. In 2019 Liu and Tarjan conjectured that an
earlier version of our algorithm takes $O(\log n)$ steps or possibly $O(\log^2
n)$ steps~\cite{bib:liu_tarjan2019}. We leave the proof of convergence as an
open problem.

\section{\label{sec:contribution}Our contribution}
We introduce a simple, deterministic label propagation method for undirected
graph connectivity. Our approach propagates directed edges in alternating
direction to achieve fast convergence independently of the number of processors
while also maintaining $O(m)$ work each step. We present new algorithms in PRAM,
Stream, and MapReduce. We will silently use the standard notation for asymptotic
bounds to provide a familiar basis for comparison, but the reader should keep in
mind that our bounds that depend on the convergence are conjecture only.

If our conjecture of $O(\log n)$ convergence is true, then our label propagation
algorithm on a Concurrent Read Concurrent Write (CRCW) PRAM achieves $O(\log n)$
time and $O(m\log n)$ work with $O(m)$ processors. On an Exclusive Read
Exclusive Write (EREW) PRAM it takes $O(\log^2 n)$ time and $O(m\log^2 n)$
work. In contrast, the fastest deterministic CRCW graph connectivity algorithms
take $O(\log n)$ time and $O((m+n)\cdot \alpha(m,n))$ work using
$O((m+n)\cdot\alpha(m,n)/\log(n))$ processors~\cite{bib:iwama1994,
  bib:colevishkin1991}, where $\alpha(m,n)$ is the inverse Ackerman
function. The best-known $\mathcal{NC}$ EREW
algorithm~\cite{bib:chonghanlam2001} takes $O(\log n)$ time and $O((m+n)\log n)$
work with $O(m+n)$ processors. Although our results are slower than those for
the fastest deterministic CRCW and EREW algorithms, our method is much simpler
and easier to implement. We also give an efficient Stream-Sort algorithm that
takes $O(\log n)$ passes and $O(\log n)$ memory, and a MapReduce algorithm
taking $O(\log n)$ rounds and $O(m\log n)$ communication overall. These would be
the first deterministic $O(\log n)$-step graph connectivity algorithms in Stream
and MapReduce models. For the purposes of this discussion we will assume $O(\log
n)$ convergence holds for our algorithm. With that assumption, refer to
Table~\ref{tab:compare} for a summary of these results and comparison to the
current state-of-the-art.

\begin{table*}
\centering
\caption{\label{tab:compare} Comparison to state-of-the-art. Cost is work,
  memory, or communication for PRAM, Stream-Sort, and MapReduce,
  respectively. (The $O(\log n)$ convergence results for this paper are
  conjecture.)}
\begin{tabular}{l|llll}
  \toprule
  Model & Steps & Cost & Class & Reference \\
  \midrule
  CRCW & $O(\log  n)$ & $O((m+n)\cdot \alpha(m,n))$ & deterministic &
  Cole, Vishkin~\cite{bib:colevishkin1991} \\

  EREW & $O(\log n)$ & $O((m+n)\log n)$ & deterministic &
  Chong, Han, and Lam~\cite{bib:chonghanlam2001} \\

  Stream-Sort & $O(\log n)$ & $O(\log n)$ & randomized &
  Aggarwal et al.~\cite{bib:aggarwal2004} \\

  MapReduce & $O(\log^2 n)$ & $O(m\log^2 n)$ & deterministic &
  Kiveras et al.~\cite{bib:kiveris2014} \\
  \midrule
  CRCW & $O(\log n)$ & $O(m\log n)$ & deterministic & This paper \\
  EREW & $O(\log^2 n)$ & $O(m\log^2 n)$ & deterministic & This paper \\
  Stream-Sort & $O(\log n)$ & $O(\log n)$ & deterministic & This paper \\
  MapReduce & $O(\log n)$ & $O(m\log n)$ & deterministic & This paper  \\
  \bottomrule
\end{tabular}
\end{table*}

The computational models we explored are briefly described in
Section~\ref{sec:model}. We refer the reader to~\cite{bib:campbell1997,
  bib:oconnell2009, bib:mcgregor2014, bib:mapreducesurvey2016} for more complete
descriptions. A survey of related work is given in
Section~\ref{sec:related}. Then in Section~\ref{sec:core} we introduce our
principal algorithm which establishes the framework behind our technique. This
leads to our PRAM results in Section~\ref{sec:pram}. In
Section~\ref{sec:extension} the framework is extended for Stream-Sort and
MapReduce models, which introduces the subject of label duplication in
Section~\ref{sec:duplication} where we identify when pathological duplication of
labels arises and how to address it. We give our Stream-Sort and MapReduce
algorithms in Sections~\ref{sec:stream} and \ref{sec:mapreduce}. Finally we
briefly describe a parallel implementation of our principal algorithm in
Section~\ref{sec:impl} followed by empirical results in Section~\ref{sec:exp}.

\section{\label{sec:model}Computational models}
In a PRAM~\cite{bib:fortune1978} each processor can access any global memory
location in unit time. Processors can read from global memory, perform a
computation, and write a result to global memory in a single clock cycle. All
processors execute these instructions at the same time. A read or write to a
memory location is restricted to one processor at a time in an Exclusive Read
Exclusive Write (EREW) PRAM. Writes to a memory location are restricted to one
processor at a time in a Concurrent Read Exclusive Write (CREW) PRAM. A
Concurrent Read Concurrent Write (CRCW) PRAM permits concurrent read and write
to a memory location by any number of processors, where concurrent writes to the
same memory location are handled by a resolution protocol. A Combining Write
Resolution uses an associative operator to combine all values in a single
instruction. A Combining CRCW employs this to store a reduction of the values,
such as the minimum, in constant time~\cite{bib:schmidt2017}.

The Stream model~\cite{bib:munropaterson1980,bib:henzinger1998} focuses on the
trade-off between working memory space $s$ and number of passes $p$ over the
input stream, allowing the computational time to be unbounded. In
W-Stream~\cite{bib:ruhl2003} an algorithm can write to the stream for subsequent
passes and in Stream-Sort~\cite{bib:aggarwal2004} the input or intermediate
output stream can also be sorted. In both W-Stream and Stream-Sort the output
streams become the input stream in the next pass. In Stream-Sort the streaming
and sorting passes alternate so a Stream-Sort algorithm reads an input stream,
computing on the items in the stream, while writing to an intermediate output
stream that gets reordered for free by a subsequent sorting pass. The streams
are bounded by the starting problem size. An algorithm in Stream-Sort is
efficient if it takes polylogarithmic passes and memory.

The MapReduce model~\cite{bib:mrc2010,bib:goodrich2011,bib:pietracaprina2012}
appeared some years after the programming paradigm was popularized by
Google~\cite{bib:mapreduce2004}. The model employs the functions \emph{map} and
\emph{reduce}, executed in sequence. The input is a set of $\langle key,value
\rangle$ pairs that are ``mapped'' by instances of the \emph{map} function into
a multiset of $\langle key,value \rangle$ pairs. The \emph{map} output pairs are
``reduced'' and also returned as a multiset of $\langle key,value \rangle$ pairs
by instances of the \emph{reduce} function. A single \emph{reduce} instance gets
all values associated with a key. A \emph{round} of computation is a single
sequence of \emph{map} and \emph{reduce} executions where there can be many
instances of \emph{map} and \emph{reduce} functions. Each \emph{map} or
\emph{reduce} function can complete in polynomial time for input $n$. Each
\emph{map} or \emph{reduce} instance is limited to $O(n^{1-\epsilon})$ memory
for a constant $\epsilon > 0$, and an algorithm is allowed $O(n^{2-2\epsilon})$
total memory. The number of machines/processors is bounded to
$O(n^{1-\epsilon})$, but each machine can run more than one instance of a
\emph{map} or \emph{reduce} function.

\section{\label{sec:related}Related work}
The famous 1982 algorithm by Shiloach and Vishkin~\cite{bib:shiloachvishkin1982}
takes $O(\log n)$ time using $O(m+n)$ processors on a CRCW PRAM, performing
$O((m+n)\log n)$ work overall. In 1991 Cole and Vishkin improved the result to
$O(\log n)$ time and $O((m+n)\cdot \alpha(m,n))$ work using
$O((m+n)\cdot\alpha(m,n)/\log(n))$ processors~\cite{bib:colevishkin1991}, but
hides a large constant in the asymptotic bound. The constant was reduced by
Iwama and Kambayashi in 1994~\cite{bib:iwama1994}. But these latter CRCW
algorithms are very complicated and difficult to translate to other
computational models because of the pointer operations. Our algorithm takes
$O(\log n)$ time using $O(m)$ processors, albeit on a more powerful CRCW. But it
is more amenable to other models because it does not rely on pointer-jumping.

The fastest deterministic EREW algorithm takes $O(\log n)$ time using $O(m+n)$
processors and is due to the 2001 work by Chong, Han, and
Lam~\cite{bib:chonghanlam2001}. This algorithm relies on carefully merging
adjacency lists, and improves the earlier 1995 result by Chong and Lam, which
took $O(\log n \log \log n)$ time. These EREW algorithm require parallel sorting
and pointer jumping. Our EREW algorithm is slower, taking $O(\log^2 n)$ time and
$O(m\log^2 n)$ work using $O(m)$ processors, but doesn't rely on pointer jumping
or sorting and is far simpler to implement.
 
The best known deterministic Stream algorithm for connected components is given
by Demetrescu, Finocchi, and Ribichini~\cite{bib:demetrescu2009}, taking
$O((n\log n)/s)$ passes and $s$ working memory size in W-Stream. Their algorithm
can only achieve $O(\log n)$ passes using $s=O(n)$ memory. A randomized
\emph{s-t}-connectivity algorithm by Aggarwal et al.~\cite{bib:aggarwal2004}
takes $O(\log n)$ passes and memory in Stream-Sort. It can be modified to
compute connected components with the same bounds, but requires sorting in three
of four steps in each pass~\cite{bib:oconnell2009}. In contrast, our Stream-Sort
connected components algorithm is deterministic and takes $O(\log n)$ passes and
$O(\log n)$ memory. It requires only one sorting step per pass and is
straightforward to implement.

A randomized MapReduce algorithm by Rastogi et al.~\cite{bib:rastogi2013} was
one of the first to show promise of fast convergence for connected
components. But it was later shown in~\cite{bib:andoni2018} that their
\emph{Hash-to-Min} algorithm~\cite{bib:rastogi2013} takes $\Omega(\log n)$
rounds. It uses a single task to send an entire component to another, which for
a giant component will effectively serialize the communication. In 2014 Kiveris
et al.~\cite{bib:kiveris2014} introduced their \emph{Two-Phase} algorithm, which
takes $O(\log^2 n)$ rounds and $O(m\log^2 n)$ communication overall. Unlike the
\emph{Hash-to-Min} algorithm it avoids sending the giant component to a single
reduce task.

We introduce a new MapReduce algorithm that is comparable to \emph{Two-Phase}
while being deterministic and simple to implement. Like \emph{Two-Phase} our
algorithm does not load components into memory or send an entire component
through a single communication channel. We go further in memory conservation by
maintaining $O(1)$-space working memory. Our MapReduce algorithm completes in
$O(\log n)$ rounds using $O(m\log n)$ communication, thereby improving the
state-of-the-art by $\Omega(\log n)$ factor in both convergence and
communication.

Although we do not study the MPC model~\cite{bib:mpc2013,bib:mpc2017} in this
paper, we want to highlight recent breakthrough work in this model. The MPC
model is a generalization of MapReduce and other Bulk Synchronous Parallel (BSP)
style models. The MPC model is more powerful than MapReduce; a MapReduce
algorithm can be simulated in MPC with the same runtime. In 2018 Andoni et
al.~\cite{bib:andoni2018} gave a randomized $O(\log D \log \log_{m/n} n)$ round
algorithm in MPC for connected components. Their algorithm uses the
leader-contraction method that works by selecting a small fraction of leader
vertices while maintaining high probability that non-leader vertices are
adjacent. To achieve this the authors add edges so the graph has uniformly large
degree, but to avoid $\Omega(n^3)$ communication cost they carefully manage how
edges are added. This result was later improved in 2019 by Behnezad et
al.~\cite{bib:behnezad2019}, who gave a $O(\log D + \log \log_{m/n} n)$-round,
randomized algorithm in MPC. More recently, Liu et al.~\cite{bib:liu2020} gave a
randomized CRCW PRAM algorithm based the work of Andoni et
al.~\cite{bib:andoni2018} and Behnezad et al.~\cite{bib:behnezad2019}, taking
$O(\log D + \log \log_{m/n} n)$ time using $O(m)$ processors. These results use
randomization and are not simple to implement, which conflicts with our
motivation.

The work most closely related to ours is that of Liu and
Tarjan~\cite{bib:liu_tarjan2019} who gave a family of label propagation
algorithms taking between $O(\log n)$ and $O(\log^2 n)$ steps. Similar to our
approach they propagate labels by minimum reduction at each step, which creates
a directed graph. But in their approach they maintain acyclicity and hence their
algorithms produce a forest of trees each step. In contrast, our algorithm does
not require maintaining trees as evident in Figure~\ref{fig:method}. They also
perform a \emph{short-cutting} operation in which the parent of every vertex is
replaced with its grandparent, whereas we only exchange labels between edge
endpoints. They had analyzed an earlier version of our algorithm and conjectured
it takes polylog steps to converge. Our current algorithm in this
paper~\cite{bib:oct192018} is simpler than our previous version and those
in~\cite{bib:liu_tarjan2019}. Moreover, it has the added benefit that the work
per step is easily shown to take no more space than the input graph.

\section{\label{sec:core}Principal algorithm}
We begin with our principal algorithm to establish the framework and core
principles. The essential operations, as succinctly summarized in the
introduction, are simply for every $(v,u)$ edge if $u$ is not $l(v)$ then
\emph{min update} and \emph{label propagation}, else \emph{symmetrization},
repeating until no label changes. Here Algorithm~\ref{alg:core} describes our
method in full. We don't specify any model now so we can focus on the basic
procedures. It should be noted that this principal algorithm achieves linear
work and fast convergence in both sequential and parallel settings. We will use
$N_k(v)$ to denote the neighborhood of a vertex $v$ at step $k$ and
$N^+_k(v)=\{v\} \cup N_k(v)$ as the closed neighborhood. Then let
$l(v)=\min(N^+_k(v))$ be the current minimum label for $v$. We use this $l(v)$
notation without a step subscript for simplicity. In our algorithm listings we
use arrays $L_{k}$ in its place, e.g. $L_k[v]$ holds $l(v)$ for $v$ at step
$k$. Only two such arrays are needed in each step. Before the algorithm starts
$L_1$ is initialized with the $l(v)$. For all algorithms we use $E_k$ to denote
the edges that will be processed at step $k$, but $E_k$ is a multiset because it
may contain duplicates.

\begin{algorithm*}[t]
\caption{\label{alg:core}}
\begin{algorithmic}[1]
\Statex $L_k$ \Comment arrays for $l(v)$ at each step $k$
\Statex Initialize $L_1$ with all starting $l(v)$
\For{$k=1,2,\ldots$ until labels converge}
  \State set $L_{k+1} := L_k$
  \ForAll{$(v,u) \in E_k$}
    \If{$u \ne L_k[v]$}
       \State set $L_{k+1}[u] := \min(L_{k+1}[u],L_k[v])$
       \State add $(u,L_k[v])$ to $E_{k+1}$ \Comment Label Propagation
    \Else
      \State add $(u,v)$ to $E_{k+1}$ \Comment Symmetrization
    \EndIf
  \EndFor
\EndFor
\end{algorithmic}
\end{algorithm*}

We employ two edge operations, \emph{label propagation} and
\emph{symmetrization}, defined as follows.

\begin{definition}
  \label{def:lp}
  Label propagation replaces $(v,u)$ with $(u,l(v))$ if $u \ne l(v)$.
\end{definition}

\begin{definition}
  \label{def:sym}
  Symmetrization replaces $(v,u)$ with $(u,v)$ if $u = l(v)$.
\end{definition}

Since $G$ does not contain loop edges then these operations cannot create loops,
otherwise it would contradict the minimum value. Thus $u = l(v)$ implies $v \ne
l(v)$. Label propagation is primarily responsible for path contraction. Minimum
label updates are concomitant with label propagation. Symmetrization keeps an
edge between a vertex and its minimum label. Thus for an edge $(v,u)$ in one
step there will be either $(u,v)$ or $(u,l(v))$ in the next step. In
Algorithm~\ref{alg:core} each edge is replaced by one new edge, either by label
propagation or symmetrization, so the edge count does not increase. The
algorithm is illustrated in Figure~\ref{fig:method}. Notice the $(2,1)$ edge is
duplicated in the second step. Also observe that the $(4,2)$ edge in the second
step produces $(2,1)$ for the third step because $L_2(4)=1$.

We remark that Algorithm~\ref{alg:core} can be terminated a number of ways
without affecting the asymptotic bounds. For example, we can detect when labels
no longer change by simply keeping a counter for the label propagation
branch. At each step the counter is set to zero. If the minimum label $l(v)$ for
a vertex $v$ is not $v$ itself, the counter is incremented. Now observe that in
the final star graph, only the root of the star can fall into the label
propagation branch but since the root is the minimum label, the counter cannot
be updated. This simple check and update takes $O(1)$ time for each label
propagation operation and is therefore free. We use this approach in our
implementation of Algorithm~\ref{alg:core} described in Section~\ref{sec:impl}.

In $G$ an undirected edge $\{u,v\}$ is comprised of counter-oriented twins
$(v,u),(u,v)$, thus there are $2m$ edges in total and $G$ is symmetric. As
stated in the introduction, we are careful to create just one twin. Thus at each
step the graph may be directed. By propagating single directed edges and
alternating the direction of the edge with a minimum label, we are able to limit
the work in each step to $O(m)$ edges while also maintaining overall
connectivity. Say for a $(v,u)$ edge that $u$ is the minimum for itself and
$v$. This $(v,u)$ becomes $(u,v)$ by symmetrization and then $(v,u)$ again by
label propagation, cycling until the algorithm ends or a new minimum is
acquired. If at some later step either endpoint gets a new minimum, then that
endpoint can propagate it to the other, thereby replacing their edge which
prohibits retaining an obsolete minimum label. In this case, $u$ cannot be
passed to $v$ again because there will be no edge to $v$.

\begin{claim}
  \label{clm:connected}
  If $(v,u)$ exists at some step, then $v,u$ are connected for all remaining
  steps in Algorithm~\ref{alg:core}.
\end{claim}

\begin{proof}
  First we will demonstrate that given $(v,u)$ in a step then $v,u$ remain
  connected in the next step because either $(u,v)$ or $(u,l(v))$ are
  created. In the former case the connection is obvious. In the latter case
  $v,u$ are connected through $l(v)$ because either $(v,l(v))$ or $(l(v),v)$
  will be simultaneously created with $(u,l(v))$.

  In the first case, $(u,v)$ can be created as follows.

  If $v=l(v)$ then $(u,v)$ is created by label propagation.

  If $u=l(v)$ then $(u,v)$ is created by symmetrization.

  In the second case $(u,l(v))$ is created by label propagation, and
  simultaneously either $(v,l(v))$ or $(l(v),v)$ will also be created as
  follows.

  In the current step there must exist $(v,l(v))$ or $(l(v),v)$ due respectively
  to label propagation or symmetrization from the previous step. Then an edge
  between $v,l(v)$ is created by the following.

  If $(v,l(v))$ exists then $(l(v),v)$ is created by symmetrization.

  If $(l(v),v)$ exists then $(v,l(v))$ is created by label propagation.

  We have now established from these two cases that $v$ and $u$ remain connected
  after one step. By inductively applying this to every new edge, it follows
  that a connection between $v$ and $u$ is preserved in all subsequent steps.
\end{proof}

\begin{lemma}
  \label{lem:correctness}
  Algorithm~\ref{alg:core} finds the connected components of $G$.
\end{lemma}

\begin{proof}
  Let $v$ be the minimum label of a component and so the minimum label for $v$
  is itself. Let $v^\prime$ denote any vertex that is not $v$ but has $v$ as its
  minimum label. Only the edges $(v,u)$ and $(u,v)$ need to be analyzed because
  all other edges will follow.

  Any $(v,u)$ is replaced with $(u,v)$ by label propagation and so $u$ gets $v$
  as its minimum label. Now $u$ is a new $v^\prime$ and will pass $v$ by label
  propagation. Thus any subsequent $(v^\prime,u)$ is the same case as $(v,u)$.

  Any $(u,v)$ where $l(u)$ is not $v$ is replaced with $(v,l(u))$. Observe that
  this $(v,l(u))$ is the same case as $(v,u)$. Moreover, any $(u,v^\prime)$
  where $l(u)$ is not $v$ is the same case as $(u,v)$, which subsequently leads
  to the case of $(v,u)$.

  Claim~\ref{clm:connected} establishes that if $(v,u)$ exists at some step,
  then $v,u$ are connected for all remaining steps. It follows that a component
  remains connected at each step. Then by induction on $(v,u),(u,v)$ cases every
  vertex in a component gets the same representative label for that component.
\end{proof}

\begin{lemma}
  \label{lem:edge_count}
  Algorithm~\ref{alg:core} creates $2m$ edges each step.
\end{lemma}

\begin{proof}
  Any new edge is a replacement of an existing $(v,u)$ edge as follows.

  If $u \ne l(v)$ then $(v,u)$ is replaced with $(u,l(v))$ by label propagation.

  If $u = l(v)$ then $(v,u)$ is replaced with $(u,v)$ by symmetrization.

  Since there are $2m$ edges in $G$ then there are $2m$ edges in the first step
  and by induction there are $2m$ edges in every step.
\end{proof}

This last result has significant practical benefit because it ensures that the
read/write costs remain linear with respect to the edges at each step, otherwise
on very large graphs these costs can be the bottleneck with respect to runtime
performance.

We have shown that our algorithm is very simple and takes $O(m)$ work each step,
making it appealing for practical applications. This is demonstrated in
Section~\ref{sec:exp} where our algorithm empirically converges in $O(\log n)$
steps on a variety of graphs. Despite the simplicity of our algorithm, the
convergence in general is difficult to analyze. Hence we can only conjecture the
following and leave the proof as a challenging open problem.

\begin{conjecture}
\label{cnj:convergence}
Algorithm~\ref{alg:core} converges in $O(\log n)$ steps.
\end{conjecture}

Instead, we can show the convergence on path graphs, which might provide insight
to a more general proof. Observe that only minimum labels are propagated and a
minimum for a vertex can only be replaced with a lesser label. At each step,
minimum labels are exchanged between the endpoints of each $(v,u)$ edge through
label propagation, and symmetrization maintains the edge between a vertex and
its minimum. Thus each length-two path is shortened by label propagation, and
symmetrization ensures that a minimum label vertex and its subordinate can both
get a new minimum from one or the other in a later step.

An interesting case is label convergence on a sequentially labeled path. On such
a path, label convergence follows a Fibonacci sequence. A na\"ive label
propagation algorithm will duplicate labels leading to a progressive increase in
the number of edges. We give the following results for our algorithm. An
interested reader can find the proofs in Appendix~\ref{sec:chain_proof}.

\begin{lemma}
  \label{lem:fibonacci}
  Algorithm~\ref{alg:core} on a sequentially labeled path propagates labels in
  Fibonacci sequence, specifically at each step $k$ the label difference
  $\Delta_k(v,l_k(v))=v-l_k(v)$ follows $\Delta_k(v,l_k(v))=F_k$ where $F_k =
  F_{k-1} + F_{k-2}$.
\end{lemma}

The label updates follow a Fibonacci sequence so the next statement holds.

\begin{proposition}
  \label{prop:fib_convergence}
  Algorithm~\ref{alg:core} converges in $\log_\phi n = O(\log n)$ steps on a
  sequentially labeled path, where $\phi \approx 1.618$ is the Golden Ratio
  value.
\end{proposition}

It is intuitive that the convergence of Algorithm~\ref{alg:core} on any path
doesn't take asymptotically more steps than on a sequentially labeled path.

\begin{claim}
  \label{clm:4path}
  Algorithm~\ref{alg:core} on any $4$-path converges in at most three steps.
\end{claim}

\begin{proof}
  Observe that every vertex in a $4$-path is within a distance of three from any
  other vertex in the path. Now recall that a vertex can only get a new minimum
  label by label propagation and label propagation contracts a length-two
  path. Let $C(G)$ denote the minimum labeled vertex of the path and is
  therefore the component label.

  It isn't difficult to see that if $C(G)$ is not an endpoint of the $4$-path
  then the algorithm converges to a star in two steps. This is because all
  vertices are a distance of two from $C(G)$ and will be connected to it in the
  first step. Then it takes one more step to break the non-tree edge from the
  endpoint that originally was not connected to $C(G)$.

  If $C(G)$ is an endpoint of the $4$-path then the other endpoint does not get
  an edge to $C(G)$ in the first step because it is a distance of three from
  $C(G)$. But this other endpoint will be connected to some other vertex that is
  connected to $C(G)$ and therefore gets $C(G)$ in the next step. Since all
  other vertices get $C(G)$ in the first step, then the total number of steps
  for convergence is three. Thus the claim holds.
\end{proof}

\begin{claim}
\label{clm:starmerge}
Given two stars rooted by their respective minimum labels and then connected by
the roots, it takes Algorithm~\ref{alg:core} at most three steps to converge.
\end{claim}

\begin{proof}
  Let $L,R$ be the respective roots of the left and right stars, and each is the
  minimum label for its star. Without loss of generality, let $L$ be less than
  $R$. Now suppose the stars are connected by either an $(L,R)$ or $(R,L)$ edge.

  If it is an $(R,L)$ edge, then symmetrization creates $(L,R)$ and label
  propagation passes $L$ to each leaf node of $R$. But also symmetrization can
  create $(R,u)$ edges from any $(u,R)$ edges. Then a subsequent step replaces
  these $(R,u)$ with $(u,L)$ edges to complete the final rooted star. This takes
  two steps in total.

  If it is an $(L,R)$ edge, then label propagation creates $(R,L)$. It then
  follows the steps for the previous case and hence takes three steps in
  total. Thus the claim holds.
\end{proof}

\begin{proposition}
\label{prop:path_convergence}
Algorithm~\ref{alg:core} converges in $O(\log n)$ steps on a path graph.
\end{proposition}

\begin{proof}
  Observe that an 8-path is just two 4-paths connected end-to-end. Let $L,R$ be
  the least smallest labels respectively for the two subpaths.

  One of these labels will be the component label. Suppose the component label
  is in the middle of the path, and therefore is the endpoint of one 4-path that
  is adjacent to the endpoint of the other 4-path. The algorithm will converge
  both 4-paths simultaneously in the same number of steps as a single
  independent 4-path. This is because minimum labels are stored externally in an
  array and it is the minimum label from this array that is propagated. Since
  both 4-paths have an endpoint whose minimum label is the component minimum,
  then both converge in the same number of steps.

  But we are interested in the worst-case. Suppose that $L,R$ are at opposite
  ends of the path, one of which is the component label. The algorithm
  simultaneously converts each to a star rooted respectively at $L$ and $R$,
  taking at most three steps by Claim~\ref{clm:4path}.

  The graph remains connected in concordance with Claim~\ref{clm:connected}, and
  because of label propagation there will be an edge connecting $L$ and $R$. It
  follows from Claim~\ref{clm:starmerge} that it takes at most three more steps
  to get the final star.

  Likewise, doubling an 8-path yields a 16-path. Then this takes at most three
  more steps to converge than the 8-path.

  A path of $n=4\cdot2^k$ size can be generated by doubling $k$ times. Each
  doubling takes at most three steps to merge connected stars according to
  Claim~\ref{clm:starmerge}. Thus it takes at most $3k$ extra steps for all the
  intermediate merges for an n-path. It follows from $n=2^{k+2}$ that $k\le \log
  n$. Then by Claim~\ref{clm:4path} it takes at most three steps to merge all
  concatenated 4-paths and $3\log n$ steps to merge intermediate stars. This
  leads to a total of $3+3\log n = O(\log n)$ steps to converge.
\end{proof}

For the remainder of this paper we will assume Conjecture~\ref{cnj:convergence}
is true.

\section{\label{sec:pram}PRAM algorithm}
Our Algorithm~\ref{alg:core} maps naturally to a PRAM. We will use the following
semantics for our parallel algorithms. All statements are executed sequentially
from top to bottom but all operations contained within a \textbf{for all}
construct are performed concurrently. All other statements outside this
construct are sequential. Recall that in a synchronous PRAM all processors
perform instructions simultaneously and each instruction takes unit time. We use
a Combining CRCW PRAM to ensure the correct minimum label is written in $O(1)$
time~\cite{bib:schmidt2017}.

\begin{theorem}
  \label{thm:crcw_time}
  Algorithm~\ref{alg:core} finds the connected components of $G$ in $O(\log n)$
  time and $O(m\log n)$ work using $O(m)$ processors on a Combining CRCW PRAM.
\end{theorem}

\begin{proof}
  For each edge it performs either label propagation or symmetrization by
  Definitions~\ref{def:lp} and~\ref{def:sym}. Each operation reads an edge from
  memory and overwrites it with a new edge. On a CRCW this takes $O(1)$ time for
  each edge. Computing the minimum on each edge also takes $O(1)$ time using
  Combining Write Resolution. In each step there are $O(m)$ edges according to
  Lemma~\ref{lem:edge_count}, and it takes $O(\log n)$ steps to converge due to
  Conjecture~\ref{cnj:convergence}. Therefore it takes $O(\log n)$ time and
  $O(m\log n)$ work given $O(m)$ processors.
\end{proof}

An EREW algorithm follows directly from Algorithm~\ref{alg:core} because it is
well-known that a CRCW algorithm can be simulated in a EREW with logarithmic
factor slowdown~\cite{bib:karp_ramachandran1990}. The only read/write conflict
in Algorithm~\ref{alg:core} is in the minimum label update. Here it does not
require a minimum reduction in constant-time. Thus for a $p$-processor EREW,
reading $L_k[v]$ takes $O(\log p)$ time by broadcasting the value in binary tree
order to each processor attempting to read it. It isn't difficult to see that a
minimum value can be found in $O(\log p)$ time using a binary tree layout to
reduce comparisons by half each step~\footnote{Given an array $M$ having $n$
  values and $p=\lceil n/2 \rceil$ processors, then at each step $M[i] \gets
  min(M[2i-1],M[2i])$ for $1 \le i \le p$, where p is halved after each
  step.}. This immediately proves Theorem~\ref{thm:erew_time}.

\begin{theorem}
  \label{thm:erew_time}
  Algorithm~\ref{alg:core} finds the connected components of $G$ in $O(\log^2
  n)$ time and $O(m\log^2 n)$ work using $O(m)$ processors on an EREW PRAM.
\end{theorem}

\section{\label{sec:extension}Extending to other models}
The Stream and MapReduce models restrict globally-shared state so the minimum
label for each vertex must be carried with the graph at each step. Recall in
Algorithm~\ref{alg:core} there may not be an explicit $(v,l(v))$ edge in a step
but all minimum labels are kept in the global $L_k$ array. So given a $(v,u)$
edge but no explicit $(v,l(v))$ edge, we can still apply label propagation and
produce $(u,l(v))$. Otherwise it would create some $(u,w)$ edge where $w$ is the
minimum of the current set of neighbors but may not be the true minimum for
$v$. This would cause the algorithm to fail, which is easily demonstrated on the
graph in Figure~\ref{fig:method}. Moreover, in MapReduce the map and reduce
functions are sequential so a giant component that is processed by one task will
serialize the entire algorithm. We address these limitations by slightly
altering the label propagation and symmetrization operations.

\begin{definition}
  \label{def:lp_ext}
  Label propagation replaces $(v,u)$ with $(u,l(v))$ if $u,v \ne l(v)$.
\end{definition}

\begin{definition}
  \label{def:sym_ext}
  Symmetrization replaces $(v,u)$ edge with $(u,v),(v,u)$ if $u = l(v)$.
\end{definition}

Label propagation and symmetrization will now only proceed from vertices $v\ne
l(v)$ to mitigate sequential processing of a giant component, thus skipping over
intermediate representative labels. As before, $u=l(v)$ implies $v\ne l(v)$,
otherwise it would contradict the minimum function. Now symmetrization adds both
edges, $(l(v),v),(v,l(v))$, so that $v$ is always paired with its minimum label
in the absence of random access to global memory. We also remark that
symmetrization must be this way when ignoring $v$ where $v \ne l(v)$ because the
vertices that would have created the edge $(v,l(v))$ are now ignored.

These minor changes do not invalidate the correctness or convergence established
by Algorithm~\ref{alg:core} because the same edges are created but with some
added duplicates. The primary difference is that both $(v,l(v)),(l(v),v)$ edges
are created in the same step rather than strided across two consecutive
steps. But this incurs label duplication that can lead to a progressive increase
in edges if left unchecked.

\section{\label{sec:duplication}Label duplication}
The new label propagation and symmetrization for Stream and MapReduce can lead
to $O(\log n)$ factor inefficiency as a result of increased label duplication,
especially on sequentially labeled path or tree graphs. This leads to the
following crucial observation.

\begin{observation}
  \label{obs:threestep}
  Adding counter-oriented edges in the symmetrization step of
  Algorithm~\ref{alg:core} will pair each $v$ with every new $l(v)$ it gets for
  three steps on a sequentially labeled path graph.
\end{observation}

\begin{proof}
  Recall from Lemma~\ref{lem:fibonacci} that $v-l_k(v) = F_k$ where $F_k=F_{k-1}
  + F_{k-2}$ is the Fibonacci recurrence. Since symmetrization retains each
  $l(v)$ for the next step then $v$ gets its $k^{th}$ label $l_k(v)=v-F_k$ for
  the next three steps because of the recurrence of $F_k$. Thus any new $l(v)$
  that $v$ receives will return to $v$ a total of three steps unless $l(v)$ is
  the minimum label for the component of $v$.
\end{proof}

Once an $l(v)$ is replaced with an updated minimum for $v$, it is no longer
needed and only adds to the edge count. Symmetrization by
Definition~\ref{def:sym_ext} retains $(v,l(v))$ so when $v$ gets $l(l(v))$ from
its current $l(v)$, then in the next step $l(l(v))$ will propagate back to
$l(v)$. Since each vertex in a chain is a minimum label to vertices up the
chain, then each vertex will in turn be back-propagated down the chain. This
follows a Fibonacci sequence hence the duplication of labels grows rapidly. For
example we can see from Lemma~\ref{lem:fibonacci} that vertex 2 in a chain will
be the minimum label for the $3,4,5,7,10,15,23,\ldots$ vertices in sequence, and
each of these vertices will return vertex 2 back to the neighbor from which it
was received. Moreover, as seen in Observation~\ref{obs:threestep} each new
$l(v)$ is retained by $v$ for three steps. Relabeling the graph can avoid the
pathological duplication but a robust algorithm is more desirable, especially in
graphs that may contain a very long chain.

Suppose now a $(u,l(v))$ edge is added to $E_{k+1}$ only if that edge is not
currently in $E_k$. This is testing if $l_k(v) \notin N_k(u)$ then it can be
added to $N_{k+1}(u)$. Since Definitions~\ref{def:lp_ext},\ref{def:sym_ext} are
applied in models with limited random access and global memory, we leverage
sorting to identify the next minimum label for each vertex and also remove
labels that would otherwise fail this membership test. Let $E'_{k+1}$ be the
intermediate edges that are created during the $k^{th}$ step and from which a
subset are retained for the $k+1$ step. Sorting edges in $E_k$ and $E'_{k+1}$
will identify those edges that are duplicated across both steps and therefore
should be removed. But an edge that is duplicated in only $E'_{k+1}$ must be
retained for proper label propagation. To avoid inadvertently removing such an
edge, all duplicates in the $E'_{k+1}$ are first removed before merging and
sorting with $E_k$ edges. After removing duplicates the edges from $E_k$ are
also removed because these were only needed for the membership test. We apply
this in our next algorithms.

\section{\label{sec:stream}Stream-Sort algorithm}
Our Stream-Sort algorithm in Algorithm~\ref{alg:streamsort_twostep} extends
Algorithm~\ref{alg:core} as described in Section~\ref{sec:extension} and removes
duplicates by sorting in the manner described at the end of
Section~\ref{sec:duplication}. It requires two stages per iteration step. The
first stage performs label propagation and symmetrization and also returns the
input edges. The second stage eliminates duplicates. A one-stage
algorithm~\cite{bib:burkhardt_graphex2016} was described in 2016, which is
simpler to implement, but does not address label duplication.

Recall that Algorithm~\ref{alg:core} does not return the current edges but
creates new edges by symmetrization and label propagation. But in Stream-Sort we
must return the current edges temporarily in the intermediate sorting stage in
order to ignore duplicate edges. Hence we mark the edges to distinguish old from
new. Here a \emph{NEW} edge resulted from either symmetrization or label
propagation, and an \emph{OLD} edge is a current input edge. Label propagation
and symmetrization follow Definitions~\ref{def:lp_ext} and~\ref{def:sym_ext}. In
the first stage Algorithm~\ref{alg:streamsort_twostep} reads sorted edges, hence
$l(v) = \min(v,u)$ from the first edge of $v$. If $u = l(v)$ for this first edge
then symmetrization adds $\bigl((u,v),NEW\bigr),\bigl((v,u),NEW\bigr)$ to an
intermediate stream $E'_{k+1}$, and for each remaining edge a
$\bigl((u,l(v)),NEW\bigr)$ is added to $E'_{k+1}$ by label propagation, and
$\bigl((v,u),OLD\bigr)$ is also added. Note that if $u=l(v)$ in the first edge
of $v$, all remaining $u$ cannot be $l(v)$ due to sorting and thus $u,v \ne
l(v)$ for the remaining edges of $v$ so label propagation can proceed. The
intermediate stream $E'_{k+1}$ is sorted by a sorting pass and input to the
second stage. In the second stage the edges are sorted so all \emph{NEW} and
\emph{OLD} versions for $(v,u)$ are grouped together. Then any edge without an
\emph{OLD} member is added to a new output stream $E_{k+1}$, which will be the
input stream in the next pass. Both intra- and inter-step duplicates have been
removed. The algorithm repeats this procedure until no new minimum can be
adopted.

\begin{algorithm*}[t]
\caption{\label{alg:streamsort_twostep}}
\begin{algorithmic}[1]
\Statex initialize $E_1$ with sorted E and set $last_v := \infty$
\For{$k=1,2,\ldots$ until labels converge}
  \For{$(v,u) \in E_k$}
    \If{$v \ne last_v$}
      \State set $l(v) := \min(v,u)$ and $last_v := v$
      \If{$u = l(v)$}
        \State add $\bigl((v,u),NEW\bigr)$ to $E'_{k+1}$
        \Comment Symmetrization
        \State add $\bigl((u,v),NEW\bigr)$ to $E'_{k+1}$
      \EndIf
    \ElsIf{$v \ne l(v)$}
      \State add $\bigl((u,l(v)),NEW\bigr)$ to $E'_{k+1}$
      \Comment Label Propagation
      \State add $\bigl((v,u),OLD\bigr)$ to $E'_{k+1}$ 
    \EndIf
  \EndFor
  \State sort $E'_{k+1}$
  \Comment{$NEW$ and $OLD$ edges get sorted together for each $(v,u)$}
  \If{$\bigl((v,u),NEW\bigr) \in E'_{k+1}$ but
    $\bigl((v,u),OLD\bigr) \notin E'_{k+1}$}
    \State add $(v,u)$ to $E_{k+1}$
  \EndIf
\EndFor
\end{algorithmic}
\end{algorithm*}

\begin{theorem}
  \label{thm:streamsort_twostep}
  Algorithm~\ref{alg:streamsort_twostep} finds the connected components of $G$
  in $O(\log n)$ passes and $O(\log n)$ memory in Stream-Sort.
\end{theorem}

\begin{proof}
  Algorithm~\ref{alg:streamsort_twostep} achieves this as follows. The input
  stream $E_k$ is scanned in one pass and the intermediate output stream
  $E'_{k+1}$ is subsequently sorted in a single sorting pass. There are then a
  constant number of sorting passes each step, which are essentially free. Label
  propagation and symmetrization extends from those in Algorithm~\ref{alg:core}
  but create more duplicate edges. These duplicates are removed in the sorting
  pass. Thus correctness follows from Lemma~\ref{lem:correctness}, there are
  $O(m)$ edges each step due to Lemma~\ref{lem:edge_count}, and it takes $O(\log
  n)$ steps to converge due to Conjecture~\ref{cnj:convergence}. The input to
  each pass is $O(m)$ thus satisfying the constraint of the Stream-Sort
  model. Sorting requires only $O(\log n)$ bits of memory to compare vertices
  and labels. Overall it takes $O(\log n)$ passes and $O(\log n)$ memory.
\end{proof}

If Conjecture~\ref{cnj:convergence} holds, this would be the first efficient,
deterministic connected components algorithm in Stream-Sort.

\section{\label{sec:mapreduce}MapReduce algorithm}
Our MapReduce algorithm described in Algorithm~\ref{alg:mapreduce_twostep} is
similar to Algorithm~\ref{alg:streamsort_twostep}, managing duplicates as
described at the end of Section~\ref{sec:duplication}. It takes two rounds per
iteration step, the first to perform label propagation and symmetrization and
the second to remove duplicates. A single-round
algorithm~\cite{bib:burkhardt_graphex2016} with less efficient communication is
also available to the interested reader. The values for each key are sorted
hence intra-step duplicates are adjacent and easily removed, permitting the
algorithm to maintain $O(1)$ working memory. Since the values are sorted then
$l(v)$ is simply the lesser between the key and first value. We omit the
specifics on sorting and getting $l(v)$ for brevity. Label propagation and
symmetrization follow Definitions~\ref{def:lp_ext} and~\ref{def:sym_ext}. The
first round returns label propagation or symmetrization edges as \emph{NEW}, and
current edges as \emph{OLD}, again to assist in removing duplicates. If $v\ne
l(v)$ then there is a $u=l(v)$, thus symmetrization emits $\langle v,(l(v),NEW)
\rangle,\langle l(v),(v,NEW) \rangle$, and for each $u\ne l(v)$ a $\langle
u,(l(v),NEW) \rangle$ is added by label propagation, and $\langle v,(u,OLD)
\rangle$ is also added. This skips any local minimum or the component minimum
which could otherwise result in a large span of sequential processing. The
second round accumulates these edges and every edge without an \emph{OLD} member
is returned with no markings. The two rounds are repeated until labels
converge. Since duplicates are removed the total communication cost is $O(m)$
per round.

\begin{algorithm*}[t]
\caption{\label{alg:mapreduce_twostep}}
\begin{algorithmic}[1]
  \Procedure{Reduce-1}{$key = v,values = N_k(v)$}
    \Comment{sorted values}
    \State set $l(v) := \min(N^+_k(v))$
    \If{$v \ne l(v)$}
      \State emit $\langle v,(l(v),NEW) \rangle$ and
      $\langle l(v),(v,NEW) \rangle$ \Comment Symmetrization
      \For{$u \in values : u \ne l(v)$}
        \State emit $\langle u,(l(v),NEW) \rangle$ \Comment Label Propagation
        \State emit $\langle v,(u,OLD) \rangle$
      \EndFor
    \EndIf
  \EndProcedure
  \Procedure{Reduce-2}{$key = v, values = \{(u,i) : i \in \{OLD,NEW\} \}$}
    \Comment{sorted values}
    \If{$\bigl((v,u),NEW\bigr) \in values$ but
      $\bigl((v,u),OLD\bigr) \notin values$}
      \State emit $\langle v,u \rangle$
    \EndIf
  \EndProcedure
\For{$k=1,2,\ldots$ until labels converge}
  \State $\text{MAP} \mapsto \text{Identity}$
  \State \Call{Reduce-1}{}
  \State $\text{MAP} \mapsto \text{Identity}$
  \State \Call{Reduce-2}{}
\EndFor
\end{algorithmic}
\end{algorithm*}

\begin{theorem}
  \label{thm:mapreduce_twostep}
  Algorithm~\ref{alg:mapreduce_twostep} finds the connected components of $G$ in
  $O(\log n)$ rounds using $O(m\log n)$ communication overall in MapReduce.
\end{theorem}

\begin{proof}
  Label propagation and symmetrization extends from those in
  Algorithm~\ref{alg:core} but create more duplicate edges. These duplicates are
  removed in each iteration step, thus correctness follows from
  Lemma~\ref{lem:correctness}. There are $O(\log n)$ steps due to
  Conjecture~\ref{cnj:convergence}, and two rounds per step for a total of
  $2\log n = O(\log n)$ rounds. The communication cost is proportional to the
  number of edges written after all rounds. Both inter- and intra-step
  duplicates are removed in each iteration step so the total number of edges
  after the second round of each step is $O(m)$ following
  Lemma~\ref{lem:edge_count}. Thus for $O(\log n)$ rounds the overall
  communication is $O(m\log n)$ as claimed.
\end{proof}

This is comparable in runtime and communication to~\cite{bib:kiveris2014}, but
differs in that it is a deterministic MapReduce algorithm for connected
components.

\section{\label{sec:impl}Implementation}
We will give some basic empirical results for our principal algorithm described
in Algorithm~\ref{alg:core}. First we'll briefly describe our parallel
implementation. We implemented the algorithm in C++ and posix threads, and for
write conflicts we used atomic operations.

There is a write conflict in updating the minimum label for each vertex. We use
the ``compare and exchange'' atomic operation to update the minimum label. But
this atomic does not test relational conditions, instead an exchange is made if
the values being compared differ. Thus to atomically update a minimum value, the
compare and exchange result must be repeatedly tested. Once a thread succeeds
with its exchange, it must test if its original data is still less than the
updated value, accomplished by a simple loop construct. Since the criteria for
updating is just a difference in value, then eventually the thread with the
minimum value will succeed and all other threads will test out. The winning
thread itself will also test out because its value will be equal to the updated
value.

Since at each step every edge is replaced by a new edge independently of other
edges, the edge list can be concurrently updated without synchronization and the
work per step is exactly $2m$ by Lemma~\ref{lem:edge_count}. An array of size
$2m$ is initialized with the input edge list and each thread is given a unique
subset of indices in this work array. Threads can then concurrently replace each
edge in their subset of the work array without conflict. Threads are blocked
until all edges are updated in a step, and then threads concurrently update the
two label arrays.

Our implementation must detect when labels no longer change. We keep a counter
for the label propagation branch to identify when label propagation no longer
updates minimum labels. At each step the counter is set to zero. If the minimum
label $l(v)$ for a vertex $v$ is not $v$ itself, the counter is incremented. In
the final star graph only the root of the star can fall into the label
propagation branch but since the root is the minimum label, the counter cannot
be updated. Specifically, in the loop over $(v,u)$ edges, we carry out label
propagation if $u$ is not $l(v)$, otherwise symmetrization is performed. We
update the counter in the label propagation branch if $l(v)\ne v$ is true. The
algorithm halts if the counter value is zero at the end of a step. Effectively
this updates the counter until each vertex is a leaf node of a star. This is
because any $v$ that is a star means that $l(v)$ is equal to $v$, hence it must
fall into the label propagation branch and then fail to update the counter.

\section{\label{sec:exp}Experiments}
We ran our parallel implementation of Algorithm~\ref{alg:core} on a workstation
with 28 Intel Xeon E5-2680 cores and 256 GB of RAM. Our experimental results are
given in Table~\ref{tbl:experiment}. The first column of
Table~\ref{tbl:experiment} lists the graphs used in our experiments. All graphs
were unweighted, without self-loops, and symmetrized with vertices labeled from
0 to $n-1$, and are therefore simple, undirected graphs with a total of $2m$
edges.

\begin{table*}[hbt!]
\caption{\label{tbl:experiment} Experimental results.}
\centering
\small
\begin{tabular}{lrrrrrr}
  \toprule
  {} & components & n (vertices) & m (edges) & D (diameter) & steps & seconds \\
  \midrule
  seqpath20 & 1 & 1,048,576 & 1,048,575 & 1,048,575 & 31 & 0.348 \\
  seqpath22 & 1 & 4,194,304 & 4,194,303 & 4,194,303 & 34 & 0.853 \\
  seqpath24 & 1 & 16,777,216 & 16,777,215 & 16,777,215 & 37 & 2.06 \\
  seqpath26 & 1 & 67,108,864 & 67,108,863 & 67,108,863 & 40 & 5.46 \\
  roadNet-PA & 206 & 1,088,092 & 1,541,898 & 787 & 17 & 0.244 \\
  roadNet-TX & 424 & 1,379,917 & 1,921,660 & 1057 & 18 & 0.295 \\
  roadNet-CA & 2638 & 1,965,206 & 2,766,607 & 854 & 17 & 0.322 \\
  road-USA & 1 & 23,947,347 & 28,854,312 & 6809 & 22 & 1.59 \\ 
  grid-0to262144by16 & 1 & 4,194,305 & 8,388,592 & 32 & 18 & 0.548 \\
  grid-0to262144by32 & 1 & 8,388,609 & 16,777,184 & 64 & 28 & 1.07 \\
  grid-0to262144by64 & 1 & 16,777,217 & 33,554,368 & 128 & 28 & 1.73 \\
  grid-0to262144by128 & 1 & 33,554,433 & 67,108,736 & 256 & 28 & 2.77 \\
  grid-0to262144by256 & 1 & 67,108,865 & 134,217,472 & 512 & 28 & 4.70 \\
  grid-0to1048576by20 & 1 & 20,971,521 & 41,943,020 & 40 & 22 & 1.57 \\
  grid-0to4194304by22 & 1 & 92,274,689 & 184,549,354 & 44 & 24 & 4.14 \\
  grid-0to16777216by24 & 1 & 402,653,185 & 805,306,344 & 48 & 26 & 14.0 \\
  web-Stanford & 365 & 281,903 & 1,992,636 & 674 & 16 & 0.191 \\
  web-BerkStan & 676 & 685,230 & 6,649,470 & 514 & 16 & 0.302 \\
  com-Youtube & 1 & 1,134,890 & 2,987,624 & 20 & 8 & 0.196 \\
  com-LiveJournal & 1 & 3,997,962 & 34,681,189 & 17 & 8 & 0.713 \\
  com-Orkut & 1 & 3,072,441 & 117,185,083 & 9 & 6 & 1.08 \\
  com-Friendster & 1 & 65,608,366 & 1,806,067,360 & 32 & 9 & 14.7 \\
  \bottomrule
\end{tabular}
\end{table*}

The first four graphs are large, sequentially labeled path graphs. Our naming
convention for these uses a suffix that denotes the number of vertices by the
base-2 exponent. Hence the graph named \emph{seqpath20} is a path of $n=2^{20}$
vertices labeled in order from $0..n-1$. It is interesting to note that the
convergence on these graphs follows the prediction asserted in
Proposition~\ref{prop:fib_convergence}. Using $\log_\phi n$ with $\phi = 1.618$,
the predicted number of steps for $n=2^{20},2^{22},2^{24},2^{26}$ is
$29,32,35,37$. Our implementation takes one extra step to test for completion
and so it closely matches the prediction.

The next four graphs are road networks with relatively large diameters. The
first three road networks are from the Stanford Network Analysis Project
(SNAP)~\cite{bib:snapnets}. The fourth road network is the USA road network from
the $9^{th}$ DIMACS Challenge~\cite{bib:dimacs2019}.

The next eight graphs are based on the example in Appendix I of Andoni et
al.~\cite{bib:andoni2018}, which was devised to be difficult for fast
convergence in graph connectivity. Each of these graphs is an $r\times c$ grid
labeled sequentially in row-major order from $1..rc$ but with a bridge node, the
zero label vertex, connecting to the first vertex of each row. Hence $n=rc+1,
m=(2r-1)c$ and the diameter is $D=2c$. We use the naming convention of
\emph{grid-0to[rows]by[columns]}. The number of rows is fixed to $r=262144$ for
the first five of these graphs. For the last three such graphs the number of
rows is $r=2^c$ and thus grows exponentially with respect to the diameter. Our
algorithm converges linearly with respect to the diameter on these last three
graphs.

The remaining graphs in the table are from SNAP~\cite{bib:snapnets}. These
graphs have small diameter, which is expected for real-world networks. Our
algorithm converges in $O(D)$ steps on these graphs.

It is evident by the number of steps in the sixth column of
Table~\ref{tbl:experiment} that our algorithm converges rapidly on these graphs,
tending towards $O(\log n)$ convergence as we conjectured. Observe that it takes
fewer steps to converge than the diameter in each of the graphs we tested,
demonstrating a real practical benefit. Moreover, the convergence rate is
independent of the number of processors. Our implementation can admit
improvement when comparing the runtime to state-of-the-art
implementations~\cite{bib:dhulipala2020}. Due to the simplicity and fast
convergence of our algorithm we believe the runtime performance can be
significantly improved. The fast convergence, simplicity, and extensibility to
other computational paradigms makes our algorithm appealing in practice.

\section*{Acknowledgments}
The author thanks David G. Harris and Christopher H. Long for their helpful
comments.

\bibliographystyle{abbrv}
\bibliography{concom}

\newpage
\appendix
\section{\label{sec:chain_proof}Sequentially labeled path}

\newtheorem*{lem:fibonacci}{Lemma~\ref{lem:fibonacci}}
\begin{lem:fibonacci}
  Algorithm~\ref{alg:core} on a sequentially labeled path propagates labels in
  Fibonacci sequence, specifically at each step $k$ the label difference
  $\Delta_k(v,l_k(v))=v-l_k(v)$ follows $\Delta_k(v,l_k(v))=F_k$ where $F_k =
  F_{k-1} + F_{k-2}$.
\end{lem:fibonacci}

\begin{proof}
  We will prove $\Delta_k(v,l_k(v)) = F_k$ by induction but first we begin with
  some preliminaries.

  Let $F_k=1,1,2,3,5,8, \ldots$ be the Fibonacci sequence over
  $k=0,1,2,3,4,\ldots,n$ and thus $F_k = F_{k-1} + F_{k-2}$. Let step $k=0$ be
  the initial state before the algorithm begins. Observe that a sequentially
  labeled path is a non-decreasing sequence, thus any subpath has the ordering
  $l(u) < u < v$. This implies $\Delta_0(v,l_0(u)) = \Delta_0(u,l_0(u)) +
  \Delta_0(v,u)$. Note that label propagation at step $k=1$ will replace
  $l_0(v)$ with $l_0(u)$, where $l_0(u)$ becomes $l_1(v)$. Hence
  $\Delta_0(u,l_0(u))$ is the same as $\Delta_1(v,l_1(v)$, leading to
  $\Delta_1(v,l_1(v)) = \Delta_0(u,l_0(u)) + \Delta_0(v,u)$. We can relabel
  $\Delta_0(u,l_0(u))$ as $\Delta_0(v,l_0(v))$ since the path is sequentially
  labeled. From this we make the assumption that $\Delta_k(v,l_k(v)) =
  \Delta_{k-1}(v,l_{k-1}(v)) + \Delta_{k-1}(v,u)$. But we want label differences
  between a vertex and its minimum label so we argue that $\Delta_{k-1}(v,u)$ is
  the same as $\Delta_{k-2}(v,l_{k-2}(v))$ using the following justification.

  It suffices to show that $\Delta_k(v,u)$ is related to an edge that is created
  by either label propagation or symmetrization. Since $F_k$ is positive then so
  is $\Delta_k(v,u)$, implying $v > u$. Recall that edges pointing from a lower
  label to higher label are due to symmetrization. Hence a $(u,v)$ edge at step
  $k$ is due to symmetrization from edge $(v,u)$ at step $k-1$, where $u$ is
  $l_{k-1}(v)$ and the label difference is the same. Then $\Delta_k(v,u)$
  can be relabeled as $\Delta_{k-1}(v,l_{k-1}(v))$. Our inductive assumption is
  now given by $\Delta_k(v,l_k(v)) = F_k = \Delta_{k-1}(v,l_{k-1}(v)) +
  \Delta_{k-2}(v,l_{k-1}(v))$ and we will prove that it holds for all steps.
 
  In the base step, we use $k=1$ and $k=2$, so we must show $\Delta_1(v,l_1(v))
  = F_1$ and $\Delta_2(v,l_2(v)) = F_2$. The first case follows trivially from
  the sequence of labels at $k=0$, thus $\Delta_1(v,l_1(v)) = 1 + 1 = 2 = F_1$
  where we use $\Delta_0(v,u)$ in place of $\Delta_{k-2}(v,l_{k-2}(v))$. For the
  second case we have $\Delta_2(v,l_2(v)) = \Delta_1(v,l_1(v)) +
  \Delta_0(v,l_0(v))$. It was established in the first case that
  $\Delta_1(v,l_1(v)) = 2$ so by substitution we get $\Delta_2(v,l_2(v)) = 2 + 1
  = 3 = F_2$.

  In the inductive step, assume $\Delta_k(v,l_k(v)) = F_k$ is valid for all
  values from one to $k$, then we must show $\Delta_{k+1}(v,l_{k+1}(v)) =
  F_{k+1}$ is also valid. This is demonstrated by,

  \begin{align*}
    \Delta_{k+1}(v,l_{k+1}(v))
    &= \Delta_{k}(v,l_k(v)) + \Delta_{k-1}(v,l_{k-1}(v)) \\
    &= F_k + F_{k-1} \\
    &= F_{k+1}.
  \end{align*}
\end{proof}

\newtheorem*{prop:fib_convergence}{Proposition~\ref{prop:fib_convergence}}
\begin{prop:fib_convergence}
  Algorithm~\ref{alg:core} converges in $\log_\phi n = O(\log n)$ steps on a
  sequentially labeled path, where $\phi \approx 1.618$ is the Golden Ratio
  value.
\end{prop:fib_convergence}

\begin{proof}
  Let $F_k = F_{k-1} + F_{k-2}$ be the $k^{th}$ number in the Fibonacci
  sequence. It follows from Lemma~\ref{lem:fibonacci} that the label updates for
  each vertex follows a Fibonacci sequence since expanding
  $\Delta_{k}(v,l_{k}(v)) = F_{k}$ leads to $l_{k}(v) = l_{k-1}(v) + l_{k-2}(v)
  - v$. Each vertex $v$ gets a new minimum label $l_{k}(v)$ at step $k$ from its
  previous minimum labels from steps $k-1$ and $k-2$ until it finally gets the
  component minimum label. Hence the $k^{th}$-labeled vertex in the path will
  get the component label $1$ in the same number of steps as it takes to get
  from 1 to $F_k$ in the Fibonacci sequence. Then it takes $\log_\phi n = O(\log
  n)$ steps for the last vertex to get label $1$, where $\phi \approx 1.618$ is
  the well-known Golden Ratio value.
\end{proof}

\end{document}